\documentclass[11pt]{amsart}
\usepackage{amsfonts,amssymb,amsmath,euscript,enumitem}

\usepackage{tikz,hyperref}
\usetikzlibrary{decorations.markings}

\theoremstyle{plain}
 \newtheorem{thm}{Theorem}[section]
 \newtheorem{prop}[thm]{Proposition}
 \newtheorem{lemma}[thm]{Lemma}
 \newtheorem{cor}[thm]{Corollary}
\numberwithin{equation}{section}

\newcommand{\mbR}{\mathbb{R}}
\newcommand{\mbC}{\mathbb{C}}
\newcommand{\mbT}{\mathbb{T}}
\newcommand{\clH}{\mathcal{H}}
\newcommand{\clK}{\mathcal{K}}

\newcommand{\Tr}{\operatorname{Tr}}
\newcommand{\supp}{\operatorname{supp}}
\renewcommand{\Im}{\operatorname{Im}}

\newcommand{\hz}{\mathfrak h_z}
\newcommand{\dom}{\operatorname{dom}} 
\newcommand{\ran}{\operatorname{ran}}
\newcommand{\scal}[1]{\left\langle #1 \right\rangle}
\newcommand{\xia}{\xi^{(a)}}
\newcommand{\xis}{\xi^{(s)}}
\newcommand{\mua}{\mu^{(a)}}
\newcommand{\mus}{\mu^{(s)}}
\newcommand{\ind}{\operatorname{ind}}
\newcommand{\eps}{\varepsilon}

\begin{document}
\title{Resonance index and singular $\mu$-invariant}
\author{Nurulla Azamov}
\author{Tom Daniels}
\address{College of Science and Engineering
   \\ Flinders University 
   \\ South Rd, Tonsley, SA 5042 Australia}
\email{nurulla.azamov@flinders.edu.au}
\email{tom.daniels@flinders.edu.au}


\begin{abstract} 
With the essential spectrum of a self-adjoint operator
given a relatively trace class perturbation 
one can associate an integer-valued invariant which admits different 
descriptions as the singular spectral shift function,
total resonance index, and singular $\mu$-invariant.
In this paper we give a direct proof of the equality of the total resonance index
and singular $\mu$-invariant assuming only the limiting absorption principle.
The proof is based on an application of the argument principle 
to the poles and zeros of the analytic continuation of the scattering matrix considered as a function of the coupling parameter.
\end{abstract}

\maketitle

\section{Introduction}
Let $\clH$ and $\clK$ be separable complex Hilbert spaces, 
$\lambda$ a real number,~$H_0$ a self-adjoint
operator on~$\clH,$ and~$V$ a symmetric form on $\clH$ such that
\begin{enumerate}
  \item[(1)] on the domain of $|H_0|^{1/2},$ $V$ admits a factorisation 
  $
    V = F^*JF,
  $ 
  where $F \colon \clH \to \clK$ is a closed operator 
  and $J$ is a self-adjoint bounded operator
  on $\clK,$
  \item[(2)] the sandwiched resolvent $T_z(H_0) = FR_z(H_0)F^*,$ 
  where $R_z(H) = (H-z)^{-1}$ is the resolvent of~$H,$ is compact 
  for some (and thus for any) $z \notin \sigma(H_0),$ the spectrum of~$H_0,$
  \item[(3)] $F$ is bounded or $H_0$ is semi-bounded,
  \item[(4)] the limit
  $T_{\lambda+i0}(H_0) := \lim_{y \to 0^+} T_{\lambda+iy}(H_0)$ exists with respect to the operator norm.
\end{enumerate}

Condition~(1) is meant in the sense that 
$
 V\colon(f,g)\mapsto \scal{Ff,JFg}
$
for any $f,g\in\dom |H_0|^{1/2}.$ 
Without loss of generality we will assume that~$F$ has trivial kernel (by extending to the kernel as any compact operator) and that $F\dom |H_0|^{1/2}$ is dense in~$\clK.$ 

Condition~(2) is meant in the sense that $T_z(H_0)$ is bounded on the domain of~$F^*$ and extends to a compact operator (denoted by the same symbol) on~$\clK.$
Combined with~(3), this implies that the perturbed operator $H_r := H_0 + rV,$ $r\in\mbR,$ is well-defined, as an operator-sum if $F$ is bounded or a form-sum if $H_0$ is semi-bounded. 

Condition~(4) is guaranteed under various additional conditions by results known collectively as the {\em Limiting Absorption Principle}, which usually refers to the fact that the set of points~$\lambda$ satisfying~(4) has full Lebesgue measure in~$\mbR.$ 

Under these conditions one can consider the following operator, which can be interpreted 
as the scattering matrix for the pair of operators~$H_0,$ $H_r,$
\begin{equation} \label{F: stationary formula}
  S(\lambda+i0; H_r,H_0) = 1 - 2ir \sqrt{\Im T_{\lambda+i0}(H_0)} J (1+rT_{\lambda+i0}(H_0)J)^{-1}\sqrt{\Im T_{\lambda+i0}(H_0)},
\end{equation}
provided the operator $1+rT_{\lambda+i0}(H_0)J$ has bounded inverse. 
By the analytic Fredholm alternative, the set of real numbers~$r$ for which the operator $1+rT_{\lambda+i0}(H_0)J$ is not invertible is discrete; elements of this set are called \emph{resonance points} (following \cite{Aza11,Aza16}).
It is a well-known fact which can also be verified by a simple calculation that for non-resonant values of the coupling parameter~$r$ the operator~\eqref{F: stationary formula} is unitary.
Moreover, for $y>0$ and any $r \in \mbR,$ the operator
\begin{equation} \label{F: stationary formula for y>0}
  S(\lambda+iy; H_r,H_0) = 1 - 2ir \sqrt{\Im T_{\lambda+i\smash[b]{y}}(H_0)} J (1+rT_{\lambda+iy}(H_0)J)^{-1}\sqrt{\Im T_{\lambda+i\smash[b]{y}}(H_0)}
\end{equation}
is also unitary and depends analytically on $y.$
Further, as $y \to +\infty,$ the operator~\eqref{F: stationary formula for y>0} converges in the uniform topology to the identity operator~$1$ and it can be shown that this convergence is locally uniform with respect to~$r$ in~$\mbR$
(Proposition~\ref{P: eigenvalues of S(z,s) converge uniformly}). 

If the value of the coupling parameter~$r$ is non-resonant then for a fixed number $e^{i\theta}$ on the unit circle $\mbT,$ one can count the number of eigenvalues of the scattering matrix $S(\lambda+iy; H_r,H_0)$ which cross the point $e^{i\theta}$ in the anticlockwise direction as~$y$ goes from $+\infty$ to~$0.$ 
Following~\cite{Pus01}, we denote this number by $\mu(\theta,\lambda; H_r,H_0)$ and call it the \emph{$\mu$-invariant}. 
The $\mu$-invariant measures the spectral flow of eigenvalues of the scattering matrix.
For relatively trace-class perturbations~$V$ the $\mu$-invariant is associated with the Lifshitz-Krein \emph{spectral shift function} (SSF) $\xi(\lambda)$ (see e.g.~\cite{Yaf92}) by the formula~\cite{Pus01} 
\[
  \xi(\lambda; H_1,H_0) = - \frac {1}{2\pi} \int_0^{2\pi} \mu(\theta, \lambda; H_1,H_0)\,d\theta, 
     \ \ \text{for  a.e.}\ \lambda\in\mbR.
\]

Sending the coupling parameter~$r$ from~$1$ to~$0$ provides another natural way to continuously deform the scattering matrix $S(\lambda+i0; H_1,H_0)$ to the identity operator~\cite{Aza11}.
Indeed, by the analytic Fredholm alternative the operator~\eqref{F: stationary formula} is a meromorphic function of the coupling parameter~$r$ considered as a complex variable.
Since this operator-function is unitary for non-resonant real~$r,$ it cannot have poles on the real axis, so that the resonant values of~$r$ are in fact removable singularities. 
In particular,
\begin{equation} \label{F: S(r) is continuous}
  \text{the map } \ [0,1] \ni r \ \mapsto \ S(\lambda+i0; H_r,H_0) \ \ \text{is continuous}. 
\end{equation}

The \emph{absolutely continuous (a.c.)~$\mu$-invariant} 
$\mua(\theta, \lambda; H_1,H_0)$ \cite{Aza11}
is the spectral flow of eigenvalues of the scattering matrix through $e^{i\theta} \in \mbT$ 
corresponding to the path~\eqref{F: S(r) is continuous}. 
This terminology is justified by the formula \cite{Aza11,AzaDan18}
\[
  \xia(\lambda; H_1,H_0) = - \frac {1}{2\pi} \int_0^{2\pi} \mua(\theta, \lambda; H_1,H_0)\,d\theta, 
      \ \ \text{for a.e.} \ \lambda\in\mbR,
\]
where $\xia(\lambda; H_1,H_0)$ is the \emph{a.c.~SSF} 
defined as the density of the a.c.~measure given by
\[
  \xia(\varphi) = \int_0^1 \Tr\left(E^{(a)}_{r}(\supp \varphi)V\varphi(H_r)\right)\,dr, \ \ \varphi \in C_c(\mbR).
\]
This formula should be compared to the Birman-Solomyak \cite{BirSol75} formula, in which $E^{(a)}_r$ is absent.
Here~$E^{(a)}_r(\cdot)$ (respectively,~$E^{(s)}_r(\cdot)$) is the spectral measure of the a.c.~(respectively, singular) part of the self-adjoint operator~$H_r.$ 

The two formulas, connecting $\xi$ and $\xia$ with $\mu$ and $\mua$ respectively, imply that
\begin{equation*}
  \begin{split}
    \xis(\lambda; H_1,H_0) & = - \frac {1}{2\pi} \int_0^{2\pi} \mus(\lambda; H_1,H_0)\,d\theta
    \\ & = - \mus(\lambda; H_1,H_0),
  \end{split}
\end{equation*}
where $\xis$ is the \emph{singular SSF} defined as the density of the a.c.~measure 
\[
  \xis(\varphi) := \int_0^1 \Tr\left(E^{(s)}_{r}(\supp \varphi)V\varphi(H_r)\right)\,dr, \ \ \varphi \in C_c(\mbR),
\]
and 
$
  \mus(\lambda; H_1,H_0) := \mu(\theta, \lambda; H_1,H_0) - \mua(\theta, \lambda; H_1,H_0)
$
is the \emph{singular $\mu$-invariant}. 
The angle variable $\theta$ is omitted from the list of arguments of the singular $\mu$-invariant 
since a simple topological argument shows that it does not depend on $\theta,$
see \cite[Section 9]{Aza11}.
This implies in particular that the singular SSF $\xis(\lambda;H_1,H_0)$ takes integer values for a.e.~$\lambda\in\mbR.$ 

One can easily check that a real number~$r_\lambda$ is a resonance point if and only if the real number $(s-r_\lambda)^{-1}$ is an eigenvalue of positive algebraic multiplicity~$N$ for the compact operator~$T_{\lambda+i0}(H_s)J$ 
for some (and hence for any) real non-resonant~$s\in\mbR.$ 
If we shift $\lambda+i0$ to $\lambda+iy,$ the eigenvalue $(s-r_\lambda)^{-1}$ changes and in general splits into eigenvalues (listed according to their multiplicities) $(s-r_{\lambda+iy}^1)^{-1}, \ldots, (s-r_{\lambda+iy}^N)^{-1}$ of the operator $T_{\lambda+iy}(H_s)J.$ 
It is well-known and not difficult to show that these shifted eigenvalues are all non-real.
Let $N_+$ and $N_-$ be the numbers of shifted eigenvalues in the upper $\mbC_+$ and lower $\mbC_-$ complex half-planes respectively.
The \emph{resonance index} \cite{Aza16,Aza17,AzaDan18} of the triple $(\lambda; H_{r_\lambda},V)$ is the number 
\[
  \ind_{res}(\lambda; H_{r_\lambda},V) := N_+-N_-.
\]
It is easily confirmed that this definition is independent of the choice of a non-resonant point~$s.$
The singular SSF $\xis(\lambda; H_1,H_0)$ obeys \cite{Aza16,AzaDan18} 
\[
  \xis(\lambda) = \sum_{r_\lambda \in [0,1]} \ind_{res}(\lambda; H_{r_\lambda},V),  \ \ \text{for a.e.} \ \lambda\in\mbR,
\]
where the sum is taken over all resonance points from $[0,1],$ of which there are a finite number.
Hence, for relatively trace-class perturbations~$V$ 
\begin{equation} \label{F: mus=sum ind(res)}
  - \mus(\lambda; H_1,H_0) = \sum_{r_\lambda \in [0,1]} \ind_{res}(\lambda; H_{r_\lambda},V), \ \ \text{for a.e.} \ \lambda\in\mbR.
\end{equation}

In \cite{AzaDan18} this equality is proved for operators $H_0$ and $H_1$ with trace class $R_z(H_0) - R_z(H_1),$
by demonstrating that both functions are equal to $\xis(\lambda)$ a.e.
In this paper, by way of proving the following theorem, we show that~\eqref{F: mus=sum ind(res)} holds
assuming only that the operators~$H_0$ and~$V$ satisfy the conditions
(1)--(4). We emphasise that these conditions do not involve any condition of trace class type,
and under such conditions the SSF and a.c.~SSF do not necessarily exist.
Having said this, it is possible that the Limiting Absorption Principle alone 
could suffice for existence of the singular SSF as was conjectured in~\cite{Aza11}. 

\begin{thm}\label{T: main thm}
For a self-adjoint operator $H_0$ and a symmetric form~$V$ obeying $(1)$--$(4)$, the equality~\eqref{F: mus=sum ind(res)} holds for any~$\lambda$ for which~$T_{\lambda+i0}(H_1)$ also exists. 
\end{thm}

The conditions (1)--(4) are well-known in scattering theory, see e.g. \cite{BirEnt67,KatKur71,Agm75,Kur76,Yaf92}. 
Three classical examples for which they hold for a.e.~$\lambda$ are: 
\begin{itemize}[leftmargin=1.6\parindent]
\item An arbitrary self-adjoint operator~$H_0$ and a trace
class self-adjoint operator~$V;$
\item A Schr\"odinger operator~$H_0 = -\Delta +
V_0(x)$ on $L_2(\mbR^\nu)$ with bounded measurable real-valued
function~$V_0(x)$ and an operator~$V$ of multiplication by a real-valued
function~$V(x)$ such that $|V(x)| \leq C
(1+|x|)^{-\nu-\eps}$ for some $C,\eps>0;$ 
\item The free Laplacian $H_0 = -\Delta$ on $L_2(\mbR^\nu)$ and 
an operator~$V$ of multiplication by a real-valued
function~$V(x)$ such that $|V(x)| \leq C
(1+|x|)^{-1-\eps}$ for some $C,\eps>0$.
\end{itemize}
In the last example the perturbation may fail to be of relatively trace class type, in which case the result of Theorem~\ref{T: main thm} is new.

We refer to the sum of resonance indices appearing on the right hand side of~\eqref{F: mus=sum ind(res)} as the {\em total resonance index}.
If the point~$\lambda$ lies outside of the common essential spectrum of the operators~$H_0$ and~$H_1,$ then the total resonance index coincides with the classical notion of spectral flow, that~is, the net number of eigenvalues which cross a given point in the positive direction.
This notion has various guises including the Fredholm index of a pair of projections as in~\cite{ASS,Phi} and the axiomatic description of~\cite{RobSal}.
Proof that the total resonance index coincides with these definitions and further information about the connection between resonance index and spectral flow can be found in~\cite{Aza17}. 
For this reason, as in~\cite{Aza16}, the total resonance index can be interpreted as an extension of the flow of singular spectrum into the essential spectrum.
The introduction of~\cite{Aza16} provides more motivation for this work. 


\medskip
{\it Acknowledgements.} We thank Prof.\ Peter Dodds for a useful discussion.

\section{Sketch of proof}
In this section we give a brief sketch of the proof.
For a fixed~$\lambda$ we consider~\eqref{F: stationary formula for y>0}
as a function, $S(\lambda+iy,r),$ of~$y$ and~$r,$ taking values in the uniform-normed group 
\[
 G(\clK) :=\{\text{invertible operators on } \clK \text{ of the form } `1 + \text{compact'}\}.
\]
There are two ways to continuously deform the scattering matrix
$S(\lambda+iy,r)|_{y=0, r=1}$ to the identity operator: by sending~$y$ from $0^+$ to $+\infty$ 
and by sending~$r$ from~$1$ to~$0.$ 
As the scattering matrix is deformed to the identity operator
its eigenvalues are also continuously deformed to~$1.$ 
The singular $\mu$-invariant is 
the difference of spectral flows through a point $e^{i\theta}$ on the unit circle, 
corresponding to these deformations. 
An attempt to deform one of these paths to another
encounters obstructions 
in the form of poles and zeros of the scattering matrix.
Overcoming these obstructions gives rise to
the resonance index $N_+-N_-$ which is the difference of the number
of poles and zeros from the upper complex half-plane, counted with their multiplicities.

We deform the path 
$
  (y \colon +\infty \rightsquigarrow 0, \ r = 1)
$ 
to the path 
$
  (y = \eps, \ r \colon 0 \rightsquigarrow 1 \ \& \ y \colon \eps \rightsquigarrow 0, r = 1),
$
as shown in Fig.~\ref{Fig: deformation} below. 
This deformation does not meet any obstructions since~$S$ is continuous 
in $[\eps,+\infty] \times [0,1]$ for $\eps>0,$ 
but it cannot be pushed to the rectangle's lower rim $y=0.$ 
However, by~\eqref{F: S(r) is continuous}, the function $S(\lambda+i0,r)$ is continuous. 
Deformation along the path~\eqref{F: S(r) is continuous} gives the $\mua$-invariant.

\begin{figure}[ht]
\begin{center}
\begin{tikzpicture}
\draw[->] (0,0) -- (7,0);
\node [right] at (7,0) {$r$};
\draw[->] (0,0) -- (0,4);
\node [above] at (0,4) {$y$};
\node [below] at (0,0) {$0$};
\node [below] at (6,0) {$1\,$};
\node [left] at (0,0) {$0$};
\node [left] at (0,3.6) {$+\infty$};

\draw (0,3.6) -- (6,3.6);
\node [above] at (3,3.6) {\tiny 1 \qquad 1 \qquad 1 \qquad 1 \qquad 1 \qquad 1 \qquad 1 };
\node [left] at (0,0.7) {\tiny 1};
\node [left] at (0,1.4) {\tiny 1};
\node [left] at (0,2.1) {\tiny 1};
\node [left] at (0,2.8) {\tiny 1};

\draw[fill] (1.8,0) circle (0.03);
\node [below] at (1.8,0) {$r_\lambda\phantom{'}$};
\draw[fill] (2.7,0) circle (0.03);
\node [below] at (2.7,0) {$r_\lambda'$};
\draw[fill] (4.5,0) circle (0.03);
\node [below] at (4.5,0) {$r_\lambda''$};

\begin{scope}[thick,decoration={
    markings,
    mark=at position 0.55 with {\arrow{latex}}}
    ] 
\draw[postaction={decorate}] (6,3.6) to (6,2);
\draw[postaction={decorate}] (6,2) to (6,0);
\draw[postaction={decorate}] (0,0.2) to (2,0.2);
\draw[postaction={decorate}] (2,0.2) to (4,0.2);
\draw[postaction={decorate}] (4,0.2) to (5.95,0.2);
\end{scope}
\draw[thick] (5.95,0) -- (5.95,0.2);
\node [above] at (3,0.2) {$y=\varepsilon$};
\draw[->] (5.8,1.5) to [bend right] (4.2,0.4); 

\draw[->] (6.8,0.7) to (6.2,0.15);
\node [above] at (7.5,0.65) {\tiny $S(\lambda+i0;H_1,H_0)$};
\end{tikzpicture}
\end{center}
\caption{The three points $r_\lambda, r'_\lambda,r''_\lambda$ represent resonance points from $[0,1]$. The $\mu$-invariant is the spectral flow of $S(\lambda+iy,r)$ along either of the paths shown.}\label{Fig: deformation}
\end{figure}
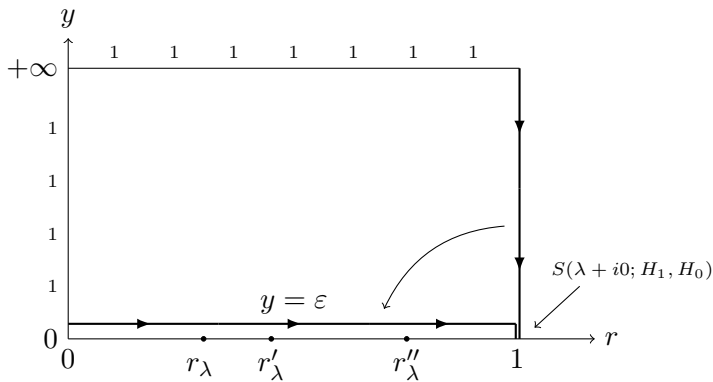

Now we fix a small value $\eps$ of~$y$ and consider $S(\lambda+i\eps,s)$ as a function of~$s\in \mbC.$ 
The factor $(1+sT_{\lambda+iy}J)^{-1}$ can have real poles only if $y=0,$ hence
such a pole~$r_\lambda = r_{\lambda+iy}|_{y=0}$ for small non-real values of~$y$
gets off the real axis. 
In doing so, it may split into finitely many poles (Fig.~\ref{Fig: splitting}), collectively known as the $r_\lambda$-group, whose multiplicities total that of $r_\lambda.$
A number $r$ is a \emph{zero} if $\dim \ker S(\lambda+iy,r)>0.$ 
Conjugates of poles of $S(\lambda+iy,s)$ are its zeros and vice versa,
as can be seen from the equality~\eqref{F: sqrt(Im T) S = M sqrt(Im T)}.

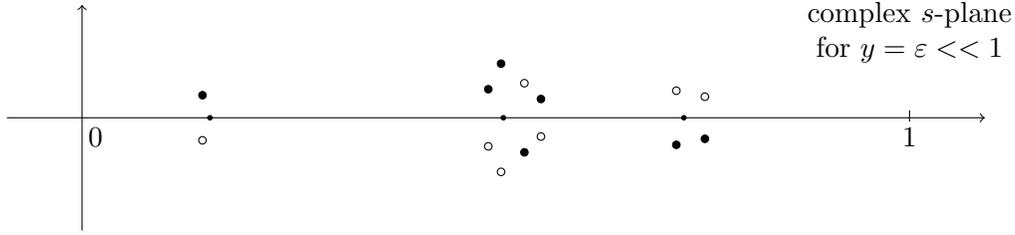
\begin{figure}[ht]
\begin{center}
\begin{tikzpicture}
\draw[->] (0,-1.5) -- (0,1.5);
\node [below] at (0.18,0) {$0$};
\draw[->] (-1,0) -- (12,0);
\node [below] at (11,0) {$1$};
\draw (11,-0.05) -- (11,0.1);
\node [above, align=center] at (11,0.6) {complex $s$-plane \\ for $y=\eps<<1$};

\draw[fill] (1.7,0) circle (0.03);
\draw[fill] (1.6,.3) circle (0.05);
\draw (1.6,-.3) circle (0.05);

\draw[fill] (5.6,0) circle (0.03);
\draw[fill] (5.4,.38) circle (0.05); 
\draw (5.4,-.38) circle (0.05);
\draw[fill] (5.57,.72) circle (0.05); 
\draw (5.57,-.72) circle (0.05);
\draw[fill] (6.1,.25) circle (0.05); 
\draw (6.1,-.25) circle (0.05);
\draw (5.88,.46) circle (0.05);  
\draw[fill] (5.88,-.46) circle (0.05);

\draw[fill] (8,0) circle (0.03);
\draw (7.9,.36) circle (0.05); 
\draw[fill] (7.9,-.36) circle (0.05);
\draw (8.28,.28) circle (0.05); 
\draw[fill] (8.28,-.28) circle (0.05);
\end{tikzpicture}
\end{center}
\caption{Black dots are poles (resonance points) and white dots are zeros (anti-resonance points).}\label{Fig: splitting}
\end{figure}

The deformation of  $S(\lambda+iy,s)$ which gives rise 
to the $\mu$-invariant is shown in Fig.~\ref{Fig: mu-invariant}.

\begin{figure}[ht]
\begin{center}
\begin{tikzpicture}
\draw[->] (0,-1.5) -- (0,1.5);
\node [below] at (0.18,0) {$0$};
\draw[->] (-1,0) -- (12,0);
\node [below] at (11,0) {$1$};
\draw (11,-0.05) -- (11,0.1);
\node [above, align=center] at (11,0.6) {complex $s$-plane \\ for $y=\eps<<1$};

\draw[fill] (1.7,0) circle (0.03);
\draw[fill] (1.6,.3) circle (0.05);
\draw (1.6,-.3) circle (0.05);

\draw[fill] (5.6,0) circle (0.03);
\draw[fill] (5.4,.38) circle (0.05); 
\draw (5.4,-.38) circle (0.05);
\draw[fill] (5.57,.72) circle (0.05); 
\draw (5.57,-.72) circle (0.05);
\draw[fill] (6.1,.25) circle (0.05); 
\draw (6.1,-.25) circle (0.05);
\draw (5.88,.46) circle (0.05);  
\draw[fill] (5.88,-.46) circle (0.05);

\draw[fill] (8,0) circle (0.03);
\draw (7.9,.36) circle (0.05); 
\draw[fill] (7.9,-.36) circle (0.05);
\draw (8.28,.28) circle (0.05); 
\draw[fill] (8.28,-.28) circle (0.05);

\begin{scope}[thick,decoration={
    markings,
    mark=at position 0.55 with {\arrow{latex}}}
    ] 
\draw[postaction={decorate}] (0,0) to (2,0);
\draw[postaction={decorate}] (2,0) to (4,0);
\draw[postaction={decorate}] (4,0) to (6,0);
\draw[postaction={decorate}] (6,0) to (8,0);
\draw[postaction={decorate}] (8,0) to (11,0);
\end{scope}
\end{tikzpicture}
\end{center}
\caption{The spectral flow of $S(\lambda+iy,s)$ along the path shown gives rise to the $\mu$-invariant.}\label{Fig: mu-invariant}
\end{figure}
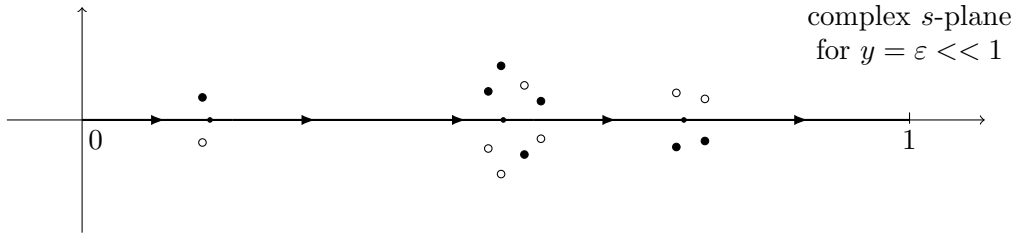

To correspond to the $\mua$-invariant, 
a deformation of $S(\lambda+iy,s)$ 
should circumvent the poles and the zeros as shown in Fig.~\ref{Fig: a.c. mu-invariant}. 
Indeed, as $y \to 0$ the $\mu$-invariant of the path in Fig.~\ref{Fig: a.c. mu-invariant}
does not change since this deformation does not encounter obstructions;
at the same time all poles and zeros of the groups of resonance points $r_\lambda, r'_\lambda,\ldots$
from $[0,1]$ converge to them where they eventually cancel each other. 
Once $y$ reaches $0$ the path below can be deformed to the straight path from $s=1$
to $s=0$ since $S(\lambda+i0,s)$ has no singularities in a neighbourhood of~$[0,1].$

\begin{figure}[ht]
\begin{center}
\begin{tikzpicture}
\draw[->] (0,-1.5) -- (0,1.5);
\node [below] at (0.18,0) {$0$};
\draw[->] (-1,0) -- (12,0);
\node [below] at (11,0) {$1$};
\draw (11,-0.05) -- (11,0.1);
\node [above, align=center] at (11,0.6) {complex $s$-plane \\ for $y=\eps<<1$};

\draw[fill] (1.7,0) circle (0.03);
\draw[fill] (1.6,.3) circle (0.05);
\draw (1.6,-.3) circle (0.05);

\draw[fill] (5.6,0) circle (0.03);
\draw[fill] (5.4,.38) circle (0.05); 
\draw (5.4,-.38) circle (0.05);
\draw[fill] (5.57,.72) circle (0.05); 
\draw (5.57,-.72) circle (0.05);
\draw[fill] (6.1,.25) circle (0.05); 
\draw (6.1,-.25) circle (0.05);
\draw (5.88,.46) circle (0.05);  
\draw[fill] (5.88,-.46) circle (0.05);

\draw[fill] (8,0) circle (0.03);
\draw (7.9,.36) circle (0.05); 
\draw[fill] (7.9,-.36) circle (0.05);
\draw (8.28,.28) circle (0.05); 
\draw[fill] (8.28,-.28) circle (0.05);

\begin{scope}[thick,decoration={
    markings,
    mark=at position 0.5 with {\arrow{latex}}}
    ] 
\draw[postaction={decorate}] (0,0) to (1.2,0);
\draw[postaction={decorate}] (1.2,0) arc (180:0:0.5);
\draw[postaction={decorate}] (2.2,0) to (4.6,0);
\draw[postaction={decorate}] (4.6,0) arc (180:0:1);
\draw[postaction={decorate}] (6.6,0) to (7.35,0);
\draw[postaction={decorate}] (7.35,0) arc (180:0:0.65);
\draw[postaction={decorate}] (8.65,0) to (11,0);
\end{scope}
\end{tikzpicture}
\end{center}
\caption{The spectral flow of $S(\lambda+iy,s)$ along the path shown gives rise to the $\mu^{(a)}$-invariant.}\label{Fig: a.c. mu-invariant}
\end{figure}
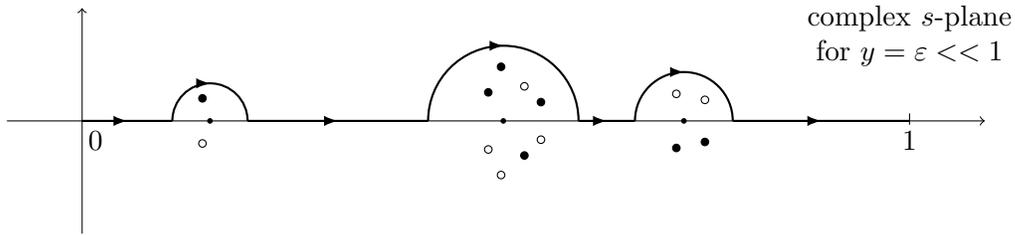

Hence, the singular $\mu$-invariant
is equal to the total number of windings of eigenvalues 
of the scattering matrix along the (anticlockwise oriented) contours
enclosing the poles and zeros of $S(\lambda+iy,s)$ from $\mbC_+$ close to the real resonance points.

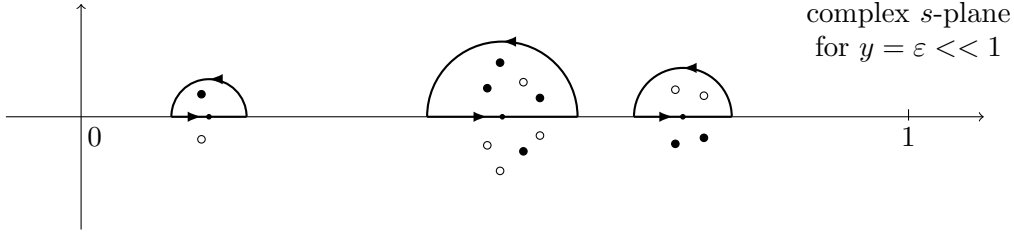
\begin{figure}[ht]
\begin{center}
\begin{tikzpicture}
\draw[->] (0,-1.5) -- (0,1.5);
\node [below] at (0.18,0) {$0$};
\draw[->] (-1,0) -- (12,0);
\node [below] at (11,0) {$1$};
\draw (11,-0.05) -- (11,0.1);
\node [above, align=center] at (11,0.6) {complex $s$-plane \\ for $y=\eps<<1$};

\draw[fill] (1.7,0) circle (0.03);
\draw[fill] (1.6,.3) circle (0.05);
\draw (1.6,-.3) circle (0.05);

\draw[fill] (5.6,0) circle (0.03);
\draw[fill] (5.4,.38) circle (0.05); 
\draw (5.4,-.38) circle (0.05);
\draw[fill] (5.57,.72) circle (0.05); 
\draw (5.57,-.72) circle (0.05);
\draw[fill] (6.1,.25) circle (0.05); 
\draw (6.1,-.25) circle (0.05);
\draw (5.88,.46) circle (0.05);  
\draw[fill] (5.88,-.46) circle (0.05);

\draw[fill] (8,0) circle (0.03);
\draw (7.9,.36) circle (0.05); 
\draw[fill] (7.9,-.36) circle (0.05);
\draw (8.28,.28) circle (0.05); 
\draw[fill] (8.28,-.28) circle (0.05);

\begin{scope}[thick,decoration={
    markings,
    mark=at position 0.5 with {\arrow{latex}}}
    ] 
\draw[postaction={decorate}] (2.2,0) arc (0:180:0.5);
\draw[postaction={decorate}] (6.6,0) arc (0:180:1);
\draw[postaction={decorate}] (8.65,0) arc (0:180:0.65);
\end{scope}
\begin{scope}[thick,decoration={
    markings,
    mark=at position 0.4 with {\arrow{latex}}}
    ]
\draw[postaction={decorate}] (1.2,0) to (2.2,0);
\draw[postaction={decorate}] (4.6,0) to (6.6,0);
\draw[postaction={decorate}] (7.35,0) to (8.65,0);
\end{scope}
\end{tikzpicture}
\end{center}
\caption{The $\mu^{(s)}$-invariant is the spectral flow of $S(\lambda+iy,s)$ along the path shown.}\label{Fig: s. mu-invariant}
\end{figure}

Combining this with the argument principle 
applied to the eigenvalues of the $S$-matrix shows that
the singular $\mu$-invariant is the number of zeros minus the number of poles inside 
of these contours, which is the (negative of the) total resonance index by definition. 

We note that it is not essential to circumvent the poles and zeros in the upper half-plane:
in this case the numbers of poles and zeros swap 
but this is compensated by the change of the contours' orientation. 

\section{Proof of Theorem \ref{T: main thm}}
Let $H_0$ and~$V=F^*JF$ be as in the introduction.
Let~$z$ belong to the resolvent set of $H_0,$
so in particular the operator $T_z(H_0)$ exists.
We denote by $S(z,s)$ the operator function
\begin{equation} \label{F: S(s,z)=I formula}
  \begin{split}
  S(z; H_s,H_0) & = 1 - 2i s \sqrt{\Im T_{z}(H_0)} J (1+sT_{z}(H_0)J)^{-1}\sqrt{\Im T_{z}(H_0)} \\
                & = 1 - 2i s \sqrt{\Im T_{z}(H_0)} (1+sJT_{z}(H_0))^{-1} J \sqrt{\Im T_{z}(H_0)},
  \end{split}
\end{equation}
where~$s \in \mbC$  and $z \in \mbC_+.$
By definition, a point~$s$ is a \emph{resonance point} or a \emph{pole}, 
respectively an \emph{anti-resonance point} or a \emph{zero}, 
corresponding to~$z,$ if 
\begin{equation} \label{F: 1+sT(z)J}
  \ker(1+sT_z(H_0)J),
\end{equation}
respectively
\begin{equation} \label{F: 1+sT(bar z)J}
  \ker(1+sT_{\bar z}(H_0)J),
\end{equation}
is non-zero. 
(The terminology ``resonance and anti-resonance points'' comes from \cite{Aza16},
here we also call them ``poles and zeros'' respectively, since they are poles and zeros 
of~\eqref{F: S(s,z)=I formula}, Corollary~\ref{C: unused}; we will use these words interchangeably).
The dimension of~\eqref{F: 1+sT(z)J} (respectively,~\eqref{F: 1+sT(bar z)J})
will be called the \emph{algebraic multiplicity} of the pole (respectively, zero).
The complex conjugate of a pole is a zero and their algebraic multiplicities 
are equal. Poles will usually be denoted by $r_z$ and zeros by $\bar r_z.$
Given $\Im z \neq 0$ we say that~$s$ is \emph{non-critical} 
if~$s$ is neither a pole nor a zero for~$z.$

\begin{lemma} For~$z$ with $\Im z>0$ and a non-critical $s$ 
\begin{equation} \label{F: sqrt(Im T) S = M sqrt(Im T)}
    \sqrt{\Im T_{z}(H_0)} S(z,s) = (1+s T_{\bar z}(H_0)J)(1+s T_{z}(H_0)J)^{-1}\sqrt{\Im T_{z}(H_0)}
\end{equation}
and
\begin{equation} \label{F: S sqrt(Im T) = sqrt(Im T) M}
    S(z,s) \sqrt{\Im T_{z}(H_0)} = \sqrt{\Im T_{z}(H_0)} (1+s J T_z(H_0))^{-1}(1+s J T_{\bar z}(H_0)).
\end{equation}
\end{lemma}
\begin{proof} Using~\eqref{F: S(s,z)=I formula}, we have
\begin{equation*}
  \begin{split}
    \sqrt{\Im T_{z}(H_0)} S(z,s) 
       & = \left[1 - s(T_{z}(H_0) -T_{\bar z}(H_0))J (1+sT_{z}(H_0)J)^{-1}\right]\sqrt{\Im T_{z}(H_0)}
    \\ & = \left[(1+sT_{z}(H_0)J)^{-1} + sT_{\bar z}(H_0)J (1+sT_{z}(H_0)J)^{-1}\right]\sqrt{\Im T_{z}(H_0)}
    \\ & =  (1 + sT_{\bar z}(H_0)J) (1+sT_{z}(H_0)J)^{-1}\sqrt{\Im T_{z}(H_0)}.
  \end{split}
\end{equation*}
The equality~\eqref{F: S sqrt(Im T) = sqrt(Im T) M} is proved similarly.
\end{proof}

\begin{lemma}\label{L: sqrt Im T_z zero kernel}
For $\Im z > 0,$ the operator $\sqrt{\Im T_z(H_0)}$ has trivial kernel.
\end{lemma}
\begin{proof}
Let $y=\Im z.$ By the first resolvent identity, 
\[
 \Im T_z(H_0) = yFR_z(H_0)(FR_z(H_0))^*.
\] 
Thus 
$
 \sqrt{\Im T_z(H_0)} = \sqrt{y}|(FR_z(H_0))^*|.
$
It follows that the kernel of $\sqrt{\Im T_z(H_0)}$ is the same as that of $(FR_z(H_0))^*,$
which is zero if~$F$ is bounded since in this case $(FR_z(H_0))^* = R_{\bar{z}}(H_0)F^*$ and both $R_{\bar{z}}(H_0)$ and $F^*$ have zero kernel. 
If $F$ is not bounded, 
then since
\[
 \ker(FR_z(H_0))^* = (\ran FR_z(H_0))^\perp,
\]
it suffices to see that $\ran FR_z(H_0) = \ran F(|H_0|+1)^{-1}$ is dense in~$\clK.$
This is true since both operators $(|H_0|+1)^{-1/2}$ and $F(|H_0|+1)^{-1/2}$ are bounded with dense range, which follows from our assumptions. 
In particular $F(|H_0|+1)^{-1/2}$ is bounded as a result of the inclusion $\dom |H_0|^{1/2} \subset \dom F$ (see e.g.~\cite[Remark~IV-1.5]{Kat80}).
\end{proof}

\begin{lemma} \label{L: S(z,s) is bdd for non-critical s}
Let $\Im z > 0.$ For any non-critical $s \in \mbC$ the operator $S(z,s)$ has a bounded inverse.
\end{lemma}

\begin{proof} By the analytic Fredholm alternative,
for non-resonant~$s$ the operator $S(z,s)$ is well-defined and bounded.
Let~$s$ be a such point.
By~\eqref{F: sqrt(Im T) S = M sqrt(Im T)} if $0\neq\varphi \in \ker S(z,s),$ 
then the vector
$
  (1+s T_{z}(H_0)J)^{-1}\sqrt{\Im T_{z}(H_0)} \varphi
$
belongs to~\eqref{F: 1+sT(bar z)J} and is non-zero by Lemma~\ref{L: sqrt Im T_z zero kernel}.
Therefore, in this case~$s$ is a zero corresponding to~$z.$
Hence, for non-critical~$s$ the operator $S(z,s)$ is bounded with zero kernel.
Since $S(z,s)-1$ is compact $S(z,s)$ 
is a bounded invertible operator.
\end{proof}

Let $\hz$ be the range of the operator $\sqrt{\Im T_{z}(H_0)}.$
As a consequence of Lemma~\ref{L: sqrt Im T_z zero kernel}, the set $\hz$ is a dense subspace of the Hilbert space~$\clK.$
For a non-critical~$s$ 
the operator $(1+s J T_{z}(H_0))^{-1} (1+s J T_{\bar z}(H_0))$ is well-defined and invertible.
Hence, by~\eqref{F: S sqrt(Im T) = sqrt(Im T) M} for such~$s$
$
  S(z,s) \hz = \hz.
$

Let
\begin{equation} \label{F: M-function}
     M(z,s) := (1+s T_{\bar z}(H_0)J)(1+s T_{z}(H_0)J)^{-1} 
            = 1 - 2is \Im T_z(H_0)(1+s T_{z}(H_0)J)^{-1}.
\end{equation}
The equality~\eqref{F: sqrt(Im T) S = M sqrt(Im T)} can be rewritten as
\begin{equation} \label{F: sqrt S = M sqrt}
  \sqrt{\Im T_{z}(H_0)} S(z,s) = M(z,s) \sqrt{\Im T_{z}(H_0)}.
\end{equation}
By Lemma~\ref{L: S(z,s) is bdd for non-critical s} 
for non-critical~$s$ the operator $S(z,s)$ is invertible, so $S(z,s)\clK = \clK.$
Hence, by~\eqref{F: sqrt S = M sqrt} for such~$s$ also $M(z,s) \hz = \hz.$
Thus, the following lemma has been proved.
\begin{lemma} \label{L: S Ez=Ez} For~$\Im z\neq 0$ and a non-critical~$s,$ $S(z,s) \hz = \hz$  and $M(z,s) \hz = \hz.$
\end{lemma}

For a given number~$z$ we say that a point~$s_0$ is a \emph{zero} of the meromorphic function~$S(z,s)$ 
iff~$0$ is an eigenvalue of~$S(z,s_0).$ It is shown below that this definition agrees with the previous one. 
The multiplicity of a zero~$s_0$ is the algebraic multiplicity of the eigenvalue~$0.$ 
An eigenvalue, distinct from~$1,$ we shall call a non-unital eigenvalue. 
By a \emph{generalised eigenvector} of an operator~$A$ corresponding to an eigenvalue~$\lambda$
we mean a non-zero solution of $(A-\lambda)^k f = 0$ for some positive $k.$ The smallest such $k$ 
is the \emph{order} of the generalised eigenvector.
\begin{lemma} Let~$z$ be a non-real complex number and~$s$ a non-critical value. Then 
(i) generalised eigenvectors (g.e.'s) of~$S(z,s)$ corresponding to non-unital eigenvalues (n.u.e.'s) belong to~$\hz.$
(ii) G.e.'s of~$M(z,s)$ corresponding to n.u.e.'s belong to~$\ran \Im T_z(H_0).$ 
(iii)~A~vector~$\varphi$ is a g.e.\ for~$S(z,s)$ corresponding to a n.u.e.\ iff $\sqrt{\Im T_z(H_0)} \varphi$ is 
a g.e.\ of the same order for~$M(z,s),$ corresponding to the same n.u.e.
\end{lemma}
\begin{proof} The first assertion follows directly from~\eqref{F: S(s,z)=I formula}.
The second assertion follows from the second of the two equalities~\eqref{F: M-function}.
The third assertion follows from~\eqref{F: sqrt S = M sqrt}, the previous two assertions 
and the fact that $\sqrt{\Im T_z(H_0)}$ has trivial kernel for non-real~$z$ by Lemma~\ref{L: sqrt Im T_z zero kernel}.
\end{proof}
\begin{cor} The operators $S(z,s)$ and $M(z,s)$ have identical spectra, including multiplicities of eigenvalues.
\end{cor}
\begin{prop} \label{P: poles and zeros of S and M} For any non-real complex number~$z$ the meromorphic functions $S(z,s)$
and $M(z,s)$ have the same sets of poles and zeros, including their multiplicities. 
\end{prop}
\begin{proof} The poles of these functions clearly coincide. 
On the other hand, $s$ is a zero of multiplicity~$N$ for $S(z,s)$ or $M(z,s)$ iff zero is an eigenvalue of $S(z,s)$ or $M(z,s)$ of multiplicity $N.$ 
The previous corollary completes the proof. 
\end{proof}

This proposition allows us to work with either $S(z,s)$ or $M(z,s)$ as far as 
we are concerned only with the eigenvalues of these operators.

\begin{cor} \label{C: unused} 
For any non-real~$z$ the meromorphic function $S(z,s)$ of~$s$ has poles at resonance points $r_z$
and zeros at anti-resonance points $\bar r_z.$ There are no other poles and zeros of $S(z,s).$ Further, 
multiplicities of a pole $r_z$ and of a zero $\bar r_z$ coincide.
\end{cor}

We define the \emph{$S$-index} of a critical point $r_z$ as the total number of (anticlockwise) windings
of eigenvalues of $S(z,s)$ around zero as~$s$ makes one winding around~$r_z.$
Due to possible branching, the contribution of individual eigenvalues to the $S$-index can be non-integral,
but the $S$-index itself is an integer. 

\begin{prop} \label{P: S-index = alg mult}
For non-real $z,$ the $S$-index of a critical point $r_z$ is~$N_2-N_1,$
where $N_1$ (respectively, $N_2$) is the algebraic multiplicity of $r_z$ as a resonance point
(respectively, as an anti-resonance point).
\end{prop}
\begin{proof} 
Let the {\em index} of a loop in $G(\clK)$ be the total number of windings around zero of its eigenvalues, which is homotopy invariant (see e.g.~\cite{ADT}). 

(A) We note that by homotopy invariance the $S$-index of a critical point $r_z$ is equal to the sum of the $S$-indices 
of resonance points into which $r_z$ splits when~$z$ is slightly perturbed.

(B) A resonance point as a function of~$z$ cannot coincide with an anti-resonance point.
Indeed, a (anti-)resonance point is a (anti-)holomorphic function of~$z,$
so they can coincide only if they are constants, which is not possible. 

To complete the proof it now suffices to prove the following statement. 

(C) Let~$z$ be a non-real number and $r_z$ a critical point 
of algebraic multiplicity~$N$ corresponding to~$z.$
If $r_z$ is not anti-resonant then its $S$-index is~$-N,$
and if $r_z$ is not resonant then its $S$-index is~$N.$

The cumbersome wording of this statement aims to avoid the possibility that a critical point is both resonant and anti-resonant. 
According to (A) and (B) above,
one can, by slightly perturbing $z,$ achieve a situation in which none of a finite number of resonance points is an anti-resonance point.

Proof. 
In the proof we assume that $r_z$ is a resonance point, the other case is similar.
By~(A) and~(B), we can also assume that $r_z$ is a non-splitting resonance point.
By the definition of a resonance point, the operator 
$
  1 + sT_z(H_0)J|_{s=r_z}
$ 
has zero as an eigenvalue, $\eps(s)|_{s=r_z},$ of multiplicity~$N.$
By assumption, $\eps(s)$ is single-valued in a neighbourhood of~$r_z.$
When~$s$ is perturbed to a value 
close to $r_z,$ the zero eigenvalue $\eps(r_z)$ shifts to $\eps(s)$ and does not split.
By the argument principle (see e.g.~\cite[\S\S\,34]{Sha92})
when~$s$ makes one winding around~$r_z$ the eigenvalue~$\eps(s)$ makes one winding around zero.
Since this eigenvalue has the same multiplicity~$N$ as~$r_z,$
the total number of windings of eigenvalues of $1+sT_z(H_0)J$ around zero
is~$N.$
Hence, the total number of windings of eigenvalues $\eps(s)^{-1}$ of the operator
$
  (1+sT_z(H_0)J)^{-1}
$
around zero, as~$s$ makes one winding around~$r_z,$ is equal to~$-N.$

Since by the premise $r_z$ is not anti-resonant the operator 
$
  1 + sT_{\bar z}(H_0)J|_{s=r_z}
$ 
is invertible.
Moreover, its loop formed as $s$ makes a winding around $r_z$ can be continuously deformed to the identity operator within~$G(\clK).$ 
Therefore as $s$ makes one winding around~$r_z,$ the index of $(1+sT_z(H_0)J)^{-1}$ is the same as that of $M(z,s).$ 
Since by Proposition~\ref{P: poles and zeros of S and M} the operators $S(z,s)$ and $M(z,s)$ have the same spectra, this completes the proof. 
\end{proof}

\begin{prop} \label{P: eigenvalues of S(z,s) converge uniformly}
For any $\lambda \in \mbR$ the operator $S(\lambda+iy,s)$
converges to the identity operator as $y \to +\infty$ locally uniformly with respect to $s \in \mbR.$
Moreover, eigenvalues $\eps_j(\lambda+iy,s)$ of the operator $S(\lambda+iy,s)$
also converge to~1 as $y \to +\infty$ locally uniformly with respect to $s \in \mbR.$
\end{prop}
\begin{proof}
By~\eqref{F: S(s,z)=I formula}, 
the first part of this proposition would be proved if we show that $T_{\lambda+iy}(H_0)$
converges to zero as $y \to +\infty.$ Combining this with unitarity of $S(z,s)$ would 
also prove the second part of the theorem.
For bounded $F$ this follows from $\|T_{\lambda+iy}(H_0)\| \leq y^{-1} \|F\|^2.$ 
Hence, by the condition (3) we can assume that $H_0$ is semi-bounded.
Since $H_0$ is semi-bounded, say lower-bounded, there is a real $m$ such that 
$R_m(H_0)$ is positive. Since $FR_m(H_0)F^*$ is compact, so is $F \sqrt{R_m(H_0)}.$
The operator $(R_{\lambda+iy}/R_m)(H_0)$ converges $*$-strongly to $0$
as $y \to \infty.$ 
Therefore its product with $F \sqrt{R_m(H_0)}$
on the left and $\big(F \sqrt{R_m(H_0)}\big)^*$ on the right converges in norm (see e.g.~\cite[Lemma 6.1.3]{Yaf92}). 
\end{proof}

\smallskip
{\it Proof of Theorem~\ref{T: main thm}.} 
For a continuous path $\gamma$ of 
operators in $G(\clK)$ which begins at the identity operator,
its $\mu$-invariant $\mu(\theta, \gamma)$ is the spectral flow of eigenvalues 
of the path through the ray $re^{i\theta},$ $r\in[0,\infty).$ 
A continuous deformation of a path $\gamma$ does not change its 
$\mu$-invariant provided that its ends stay fixed.
(The homotopy invariance of spectral flow is well-known. 
Our favoured interpretation in this context is based on continuous enumeration of the eigenvalues, see e.g.~\cite{ADT}.)
For $y>0$ and real~$s$ the operator $S(\lambda+iy,s)$ 
is continuous in the rectangle $[y_0,Y_0] \times [0,1]$
for any $0 < y_0 < Y_0 < +\infty.$ 
Proposition~\ref{P: eigenvalues of S(z,s) converge uniformly} shows 
that $S(\lambda+iy,s)$ is continuous in the larger rectangle 
$
  [y_0,+\infty] \times [0,1].
$
Hence, the $\mu$-invariant $\mu(\theta,\lambda; H_1,H_0)$ equals 
the $\mu$-invariant of the path 
$
  (0,1) \to (y_0,1) \to (y_0,0)
$ 
for any $y_0>0,$ see Fig.~\ref{Fig: deformation}. 
The difference between the $\mu$-invariants of this path, see Fig.~\ref{Fig: mu-invariant},
and of the one which 
circumvents the critical points of the group of~$r_\lambda$ 
from above, see Fig.~\ref{Fig: a.c. mu-invariant}, is equal to the sum of the $S$-indices of the points from 
the upper half-plane, Fig.~\ref{Fig: s. mu-invariant}. 
By Proposition~\ref{P: S-index = alg mult} 
this sum of the $S$-indices is equal to the negative of the resonance index of~$r_\lambda.$

Further, the $\mu$-invariant of the path which circumvents 
the critical points from above does not change as $y \to 0^+$
(the critical points of the group of~$r_\lambda$ converge 
to~$r_\lambda$ but this point is circumvented).
The path thus obtained can further be continuously deformed to the straight 
line path from $s=1$ to $s=0$ in the complex plane, since when $y=0$
the function $S(\lambda+iy,r)$ is holomorphic in a neighbourhood of the real axis, 
by~\eqref{F: S(r) is continuous}.
Hence, the $\mu$-invariant of this path is equal to the absolutely continuous 
part of the $\mu$-invariant.
This completes the proof.

\end{document}